%% file: paper.tex
\documentclass[11pt]{article}
\usepackage[latin1]{inputenc}
\usepackage{graphicx}
\usepackage{epsfig}
\usepackage{color}
\usepackage{a4wide}
\usepackage{multirow}
\usepackage{lscape}
\usepackage{verbatim}
\usepackage{amsthm, amssymb}
\usepackage{amsmath}
\newtheorem{Lem}{Lemma}
\newtheorem{theorem}{Theorem}

\def\rank{\operatorname{rank}}
\def\edge{\operatorname{edge}}
\def\branch{\operatorname{branch}}
\def\induced{\operatorname{induced}}

\title{Faster Deterministic Fully-Dynamic Graph Connectivity}
\author{Christian Wulff-Nilsen
        \footnote{Department of Computer Science,
                  University of Copenhagen,
                  \texttt{koolooz@diku.dk},
                  \texttt{http://www.diku.dk/$_{\widetilde{~}}$koolooz/}.}}

\date{}
\begin{document}

\maketitle
\begin{abstract}
We give new deterministic bounds for fully-dynamic graph connectivity.
Our data structure supports updates (edge insertions/deletions) in
$O(\log^2n/\log\log n)$ amortized time and connectivity queries in
$O(\log n/\log\log n)$ worst-case time, where $n$ is the number of vertices of the graph.
This improves the deterministic data structures of Holm, de Lichtenberg, and Thorup
(STOC 1998, J.ACM 2001) and Thorup (STOC 2000) which both have $O(\log^2n)$ amortized
update time and $O(\log n/\log\log n)$ worst-case query time. Our model of computation
is the same as that of Thorup, i.e., a pointer machine with standard $AC^0$ instructions.
\end{abstract}
\newpage

\section{Introduction}\label{sec:Intro}
The dynamic graph connectivity problem is perhaps the most fundamental
dynamic graph problem and has received considerable attention from the
algorithms community for decades. The goal is to build an efficient
data structure that supports one or more of the following operations in a
dynamic graph $G$:
\begin{itemize}
\item \texttt{connected}$(u,v)$: determines whether vertices $u$ and $v$ are connected in $G$,
\item \texttt{insert}$(u,v)$: inserts edge $(u,v)$ in $G$,
\item \texttt{delete}$(u,v)$: deletes edge $(u,v)$ from $G$.
\end{itemize}
The fully-dynamic graph connectivity problem supports all three operations, whereas the
simpler decremental and incremental variants do not support \texttt{insert} and \texttt{delete},
respectively. In the following, we refer to both \texttt{insert} and \texttt{delete} as update
operations.

The first non-trivial data structure for fully-dynamic graph connectivity is due to
Frederickson~\cite{Frederikson} who showed how to support updates in $O(\sqrt m)$ time
and connectivity queries in $O(1)$ time, where $m$ is the number of edges of the graph.
Using a general sparsification technique, Eppstein, Galil, Italiano, and
Nissenzweig~\cite{Sparsification} improved
update time to $O(\sqrt n)$, where $n$ is the number of vertices. Both of these
data structures are deterministic and the time bounds are worst-case.

Henzinger and King~\cite{HK95} significantly improved update time to $O(\log^3n)$ with
only a small increase in query time to $O(\log n/\log\log n)$. However, their bounds are
randomized expected and update time is amortized. Using a clever sampling technique, Henzinger and
Thorup~\cite{HT97} shaved off a factor of $\log n$ in the update time. A simple
and elegant deterministic data structure with the same bounds as in~\cite{HT97} was
given by Holm, de Lichtenberg, and Thorup~\cite{HLT01}. Its space requirement was later improved to linear
by Thorup~\cite{Thorup00} who also gave a randomized data structure with
a faster update time of $O(\log n(\log\log n)^3)$ and marginally slower query time of
$O(\log n/\log\log\log n)$.

A general cell-prove lower bound of $\Omega(\log n)$ for fully-dynamic graph connectivity
was provided by P\u{a}tra\c{s}cu and Demaine~\cite{PatrascuDemaine}. Hence, the data
structures above are near-optimal.

As shown by Tarjan~\cite{Tarjan}, incremental connectivity is the union-find problem
which can be solved in $O(\alpha(m,n))$ time over $m$ updates. Thorup~\cite{Thorup97}
gave an $O(\log n)$ bound for decremental connectivity if the initial graph has
$\Omega(n\log^5n)$ edges. He also gave an $O(1)$ bound if the initial graph is dense.

Our contribution is a deterministic data structure that improves the update time of
the deterministic data structures in~\cite{HLT01,Thorup00} by a factor of $\log\log n$.
We use several ingredients of Thorup~\cite{Thorup00}, including his structural forest
(which we refer to as a cluster forest) as well as lazy local trees and shortcuts both of
which we modify to suit our needs. We also introduce an additional system of shortcuts that
allows us to more quickly walk up trees of the cluster forest. Together, these changes and additions give an order $\log\log n$ speed-up in the update time.

Table~\ref{tab:DataStructs} summarizes the results for fully-dynamic graph connectivity.

\begin{table*}
\begin{center}
\begin{tabular}{|c|c|c|c|}
\hline
Update time & Query time & Type & Reference\\
\hline
$O(\sqrt m)$ & $O(1)$ & deterministic; worst-case time & ~\cite{Frederikson}\\
\hline
$O(\sqrt n)$ & $O(1)$ & deterministic; worst-case time & ~\cite{Sparsification}\\
\hline
$O(\log^3n)$ & $O(\log n/\log\log n)$ & randomized; amortized update time& ~\cite{HK95}\\
\hline
$O(\log^2n)$ & $O(\log n/\log\log n)$ & randomized; amortized update time & ~\cite{HT97}\\
\hline
$O(\log^2n)$ & $O(\log n/\log\log n)$ & deterministic; amortized update time & ~\cite{HLT01,Thorup00}\\
\hline
$O(\log n(\log\log n)^3)$ & $O(\log n/\log\log\log n)$ & randomized; amortized update time & ~\cite{Thorup00}\\
\hline
$O(\log^2n/\log\log n)$ & $O(\log n/\log\log n)$ & deterministic; amortized update time & This paper\\
\hline
\end{tabular}
\end{center}
\caption{Performance of data structures for fully-dynamic graph connectivity.}\label{tab:DataStructs}
\end{table*}

Our paper is organized as follows. In Section~\ref{sec:DefsNot}, we introduce basic definitions and notation.
Section~\ref{sec:SimpleDataStruct} gives a simple data structure with $O(\log^2n)$ update time and $O(\log n)$
query time. It is essentially the deterministic data structure of Thorup~\cite{Thorup00} but a slightly more minimalistic
variant that does not need to maintain spanning trees of connected components. In Section~\ref{sec:ImprovedDataStruct},
we add two systems of shortcuts to our data structure that together improve both update and query time by a factor
of $\log\log n$. The simplification given in Section~\ref{sec:SimpleDataStruct} is not important in order to
get our improvement in Section~\ref{sec:ImprovedDataStruct}; indeed, our result can easily be extended to maintain
a spanning forest. However, we believe that our approach gives a slightly cleaner analysis
and it should give a small improvement in practice. For instance, Thorup's data structure needs to maintain two
types of bitmaps for edges, one for tree edges and one for non-tree edges whereas our data structure only needs to maintain
one type; see Section~\ref{subsec:SearchingEdges} for details. Finally, we give some concluding remarks in
Section~\ref{sec:ConclRem}.

\section{Preliminaries}\label{sec:DefsNot}
Let $\log$ denote the base $2$ logarithm.
We assume the same model of computation as in~\cite{Thorup00}, i.e., a pointer machine with
words (bitmaps) containing at least $\lfloor\log n\rfloor + 1$ bits and with the following
standard $AC^0$ instructions: addition, subtraction, bitwise 'and', 'or', and 'not', and bit shifts. Our
data structure also needs to perform division $x/y$ and multiplication $xy$ which are \emph{not} $AC^0$ instructions.
To handle this, we assume that $y$ is a power of $2$ so that a bit shift operation can be used instead;
we can always round $y$ to the nearest such value and the constant multiplicative error introduced will not
affect correctness or running time. For a bitmap $m$, we denote by $m[i]$ the $i$th bit of $m$, $i\geq 0$.

We let $G$ denote the input graph and it is assumed to contain no edges initially. To distinguish between vertices of
$G$ and other vertices (such as those in trees of our data structure), we refer to the latter as nodes.
For a path $P$ and nodes $a,b\in P$, $P[a,b]$ is the subpath of $P$ between $a$ and $b$.
We abbreviate balanced binary search tree as BBST and depth-first search as DFS.

\section{A Simple Data Structure}\label{sec:SimpleDataStruct}
We first give a simplified version of our data structure which is no
better than the deterministic data structures in~\cite{HLT01} and~\cite{Thorup00}. In fact, it has a slower query time
of $O(\log n)$. In Section~\ref{sec:ImprovedDataStruct}, we shall speed up both query and update time by a factor of
$\log\log n$.

\subsection{The cluster forest}\label{subsec:ClusterForest}
As in~\cite{Thorup00}, we assign to each edge $e$ of $G$ a \emph{level} $\ell(e)$ between $0$ and
$\ell_{\max} = \lfloor\log n\rfloor$ and for $0\leq i\leq l_{\max}$, we denote by $G_i$ the subgraph of $G$ induced
by edges $e$ with $\ell(e)\geq i$. We refer to the connected components of $G_i$ as \emph{level $i$ clusters} or just
\emph{clusters}. The following invariant will be maintained by our data structure:
\begin{description}
\item [Invariant:] For each $i$, any level $i$ cluster spans at most $\lfloor n/2^i\rfloor$ vertices.
\end{description}

The \emph{cluster forest} of $G$ is a forest $\mathcal C$ of rooted trees where each node $u$ corresponds to a cluster $C(u)$.
The \emph{level} $\ell(u)$ of $u$ is its depth in $\mathcal C$ (between $0$ and $\ell_{\max}$) and a level $i$ node
corresponds to a level
$i$ cluster. In particular, roots of $\mathcal C$ correspond to level $0$ clusters and hence to the connected components
of $G_0 = G$. By the invariant, each leaf of $\mathcal C$ corresponds to a vertex of $G$; we often identify the two
and our data structure keeps bidirected pointers between them.
A node $u$ at a level $i < \ell_{\max}$ has as children the level $(i+1)$ nodes $v$ such that $C(v)\subseteq C(u)$.

Our data structure will maintain, for each node $u$ of $\mathcal C$, an integer $n(u)$ denoting the number of
leaves in the subtree of $\mathcal C$ rooted at $u$. In other words, $n(u)$ is the number of vertices of $G$ spanned by
$C(u)$. This completes the description of the cluster forest.

Given $\mathcal C$, we can determine whether two vertices $u$ and $v$ are connected in $G$ in $O(\log n)$ time as follows.
Traverse paths from $u$ and $v$ to roots $r_u$ and $r_v$ of trees of $\mathcal C$ containing $u$ and $v$, respectively.
Then $u$ and $v$ are connected iff $r_u = r_v$.

\subsection{Handling insertions}\label{subsec:Insertion}
We need to maintain $\mathcal C$ as $G$ changes. First we describe how to update $\mathcal C$ after an operation
\texttt{insert}$(u,v)$. We initialize $\ell(u,v)\leftarrow 0$. Letting $r_u$ and $r_v$ be defined as above, if
$r_u = r_v$, no
update of $\mathcal C$ is required since $u$ and $v$ were already connected in $G = G_0$.
If $r_u\neq r_v$, we update $\mathcal C$ by \emph{merging $r_u$ and $r_v$ into $r_u$},
meaning that $r_u$ inherits the children of $r_v$, $n(r_u)$ is increased by $n(r_v)$, and $r_v$ is deleted. This update
corresponds to merging $C(r_u)$ and $C(r_v)$. Thus $\mathcal C$ is correctly updated and the invariant still holds.

\subsection{Handling deletions}\label{subsec:Deletion}
Now consider the update \texttt{delete}$(u,v)$. Let $i = \ell(u,v)$ and let $C_u$ and $C_v$ be the
level $(i+1)$ clusters containing $u$ and $v$, respectively. Assume that $C_u\neq C_v$ since otherwise, the connectivity
in $G_i$ is not affected (there is a $uv$-path in $G_{i+1}\subset G_i$ connecting $u$ and $v$).
Let $M_i$ be the multigraph with level $(i+1)$ clusters as vertices and with level $i$-edges of
$G$ as edges (so an edge of $M_i$ connects two vertices if that edge connects the corresponding level $(i+1)$ clusters
in $G_i$). Our algorithm will not actually keep $M_i$ but it will help to simplify the description in this subsection.

We now execute two standard search procedures (say, DFS) in $M_i$, one, $P_u$, starting from vertex $C_u$ and
another, $P_v$, starting from $C_v$. The two procedures are executed in ``parallel'' by alternating between the two,
i.e., one unit of time is spent on $P_u$, then one unit on $P_v$, then one unit on $P_u$, and so on. We terminate
both procedures as soon as we are in one of the following two cases (which must happen at some point):
\begin{enumerate}
\item a vertex of $M_i$ is explored by both procedures,
\item one of the procedures has no more edges to explore and we are not in case $1$.
\end{enumerate}

In the following, we show how to deal with these two cases.

\paragraph{Case $1$:}
Let $C_{uv}$ be the vertex (level $(i+1)$ cluster) of $M_i$ explored by both procedures. Assume w.l.o.g.~that $P_v$ was the
last to explore $C_{uv}$. Let $\mathcal C_u$ be the set of level $(i+1)$ clusters explored by $P_u$ and let $\mathcal C_v$
be the set of level $(i+1)$ clusters explored by $P_v$, excluding $C_{uv}$. If we let
$n_u = \sum_{C\in\mathcal C_u}n(C)$ and $n_v = \sum_{C\in\mathcal C_v}n(C)$ then since
$\mathcal C_u\cap\mathcal C_v = \emptyset$, we have $n_u + n_v\leq\lfloor n/2^i\rfloor$ by our invariant and thus
$\min\{n_u, n_v\}\leq \lfloor n/2^{i+1}\rfloor$.

Assume first that $n_u\leq n_v$. Then we can increase the level of every edge explored by $P_u$
from $i$ to $i+1$ without violating the invariant. To see this, note that the level updates correspond
to merging clusters of $\mathcal C_u$ into one level $(i+1)$ cluster spanning
$n_u\leq\lfloor n/2^{i+1}\rfloor$ vertices. The idea is that the search performed by $P_u$ is paid for by these level
increases. As $P_v$ spent the same amount of time as $P_u$ (up to an additive $O(1)$ term), the level increases also pay
for the search by $P_v$.

We need to update cluster forest $\mathcal C$ accordingly. When increasing the level of an edge $e$ from $i$ to $i+1$,
we identify the level $(i+1)$-ancestors $a$ and $b$ of the endpoints of $e$. Clusters $C(a)$ and $C(b)$ need to be
merged and we do this by merging $a$ and $b$ into $b$ and updating $n(a)\leftarrow n(a) + n(b)$.
As we will see later, this update can also be paid for by the level
increase of $e$. Note that the procedures have found a
replacement path in $G_i$ for deleted edge $(u,v)$ so no further updates are required in $\mathcal C$, and we terminate.

We assumed above that $n_u\leq n_v$. We do exactly the same for clusters in $\mathcal C_v$ when
$n_v\leq n_u$ except that we do not increase the level of the last edge explored by $P_v$ as it
connects to a cluster in $\mathcal C_u$. If this was the only edge explored, there are no edges to pay for it but
in this case we have found a replacement path for edge $(u,v)$ and the \texttt{delete}$(u,v)$-operation can pay.

\paragraph{Case $2$:}
Now assume that one of the procedures, say $P_u$, explores all edges in the connected component of $M_i$ containing
$C_u$ and that we are not in case $1$.
Let us assume that $n_u\leq\lfloor n/2^{i+1}\rfloor$; if not, we fully explore the connected component of
$M_i$ containing $C_v$ and update $n_v$ which will be at most $\lfloor n/2^{i+1}\rfloor$ by our invariant; the
description below then applies if we swap the roles of $u$ and $v$.

We can conclude that no replacement path for $(u,v)$ exists in $G_i$.
All edges explored by $P_u$ have their level increased to $i+1$ and we update $\mathcal C$ accordingly by
merging all level $(i+1)$ nodes explored by $P_u$ into one, $w$, and setting $n(w)$ to the sum of $n(w')$ for
nodes $w'$ explored by $P_u$. These level increases pay for the two searches.
Since $C_u$ and the component of $M_i$ containing $C_v$ are no longer connected in
$G_i$, we further update $\mathcal C$ as follows: let $p$ be the parent of $w$ in $\mathcal C$. We remove $w$ as a
child of $p$, decrease $n(p)$ by $n(w)$, add $w$ as a child of a new level $i$ node $p'$, set $n(p') = n(w)$, and
add $p'$ as a child of the parent of $p$. This correctly updates $\mathcal C$ and the invariant is maintained.

If $i > 0$, it may still be possible to reconnect $u$ and $v$ in $G_j$ for some $j < i$.
We thus execute the above algorithm recursively with
$i\leftarrow i - 1$, $C_u\leftarrow C(p')$, and $C_v\leftarrow C(p)$. Should we end up in case $2$ with
$i = 0$, no replacement path in $G$ between $u$ and $v$ could be found. Then $p'$ becomes a new root of
$\mathcal C$ and we terminate.

\subsection{Local trees}\label{subsec:LocalTrees}
We now extend our data structure to allow the search procedures to explore edges and vertices of multigraph $M_i$ in a
more efficient way.

First we shall convert $\mathcal C$ into a forest of binary trees by adding \emph{local trees} as in~\cite{Thorup00}.
Let $u$ be a non-leaf node of $\mathcal C$.
We form a binary tree $L(u)$ with $u$ as root and with the children of $u$ as leaves as follows. First assign a
\emph{rank} $\rank(v)\leftarrow \lfloor\log n(v)\rfloor$ to each child $v$ of $u$. Initially, each such $v$ is
regarded as a tree consisting just of $v$. While there are trees $T$ and $T'$ whose roots $r$ and $r'$ have the
same rank, we pair them by attaching $r$ and $r'$ to a new root $r''$ with $\rank(r'')\leftarrow\rank(r) + 1$.
We end up with at most $\log n$ trees $T_1,\ldots,T_k$, called \emph{rank trees}, whose roots $r_1,\ldots,r_k$ have
pairwise distinct ranks: $\rank(r_1) > \rank(r_2) > \cdots > \rank(r_k)$. We connect the rank trees into a single
\emph{local tree} $L(u)$ rooted at $u$ by adding a \emph{rank path} $v_1v_2\ldots v_{k-1}$ down from $v_1 = u$ and
connecting $r_i$ as a child to $v_i$ for $i = 1,\ldots,k-1$ and $r_k$ as a child to $v_{k-1}$.
We define $\rank(u)\leftarrow\lfloor\log n(u)\rfloor$.

The edges in $\mathcal C$ from $u$ to its children are replaced by local tree $L(u)$; let $\mathcal C_L$ be $\mathcal C$ after
all local trees have been added. As shown by
Thorup, for a child $v$ of a node $u$ in $\mathcal C$, the depth of $v$ in $L(u)$ is at most $\log(n(u)/n(v)) + 1$.
Since any leaf of $\mathcal C$ has depth at most $\ell_{\max} = \lfloor\log n\rfloor$, a telescoping sums argument
implies that any leaf of $\mathcal C_L$ has depth $O(\log n)$.

Refer to nodes of $\mathcal C_L$ that are also nodes of $\mathcal C$ as \emph{$\mathcal C$-nodes}. Our data structure
will maintain $\mathcal C_L$ as well as $n(u)$ for each $\mathcal C$-node $u$ and $\rank(v)$ for each node
$v\in\mathcal C_L$.

\subsection{Searching for edges}\label{subsec:SearchingEdges}
We shall use $\mathcal C_L$ to search for edges in $M_i$. To facilitate this, we associate a bitmap $\edge(u)$ with each
node $u$ of $\mathcal C_L$ where $\edge(u)[i] = 1$ iff a level $i$-edge is incident to
a leaf of the subtree of $\mathcal C_L$ rooted at $u$.\footnote{Thorup's data structure needs two bitmaps in order to
distinguish between tree edges and non-tree edges whereas we only need one; $\edge(u)$ can be regarded as the bitwise
'or' of his two bitmaps.}

We can use these bitmaps to search for the edges of $M_i$. Consider one of the search procedures, say $P_u$, described
above. At any point in the search, a set of vertices of $M_i$ have been explored and these correspond to level
$(i+1)$ nodes in $\mathcal C_L$ that we mark as explored. With the bitmaps, we identify unexplored descendant leaves of
marked nodes in $\mathcal C_L$
that are incident to level $i$-edges and hence to edges of $M_i$ that should be explored by $P_u$. At each leaf, we have
all incident edges grouped according to their level. A BBST is kept which allows us to get down to
a particular group in $O(\log\ell_{\max}) = O(\log\log n)$ time. When a level $i$-edge $(a,b)$ is explored in the
direction from $a$ to $b$, we determine the endpoint in $M_i$ corresponding to $b$ by moving up from leaf $b$ to the
ancestor level $(i+1)$ node in $\mathcal C_L$. Finally, we mark this level $(i+1)$ node as explored.
Since $\mathcal C_L$ has $O(\log n)$ height, we can execute $P_u$ in $O(\log n)$ time per edge explored.

\subsection{Maintaining $\mathcal C_L$}\label{subsec:MaintainCL}
We now describe how to maintain $\mathcal C_L$ as $\mathcal C$ is updated. Let us consider the update in $\mathcal C$
of merging nodes $u$ and $v$ into $u$. In $\mathcal C_L$, this is done by first removing the rank paths in
$L(u)$ and $L(v)$, leaving at most $\log n$ rank trees of distinct
ranks for each of the nodes $u$ and $v$. We may assume that rank trees are kept in two lists sorted by the ranks of their
roots and we merge the two lists into one and start pairing up trees whose roots have the same rank, in the same way
as above. We connect their roots with a new rank path, thereby creating the new $L(u)$ and we identify its root with
$u$. Total time for a merge is $O(\log n)$.

We also need to update $\mathcal C_L$ when a child $b$ in $\mathcal C$ is added to or removed from a node $a$ (we need
this when failing to find a replacement path at some level). If $b$ is to be added, we can regard it as a trivial rank
tree that should be added to $L(a)$. This can be done in $O(\log n)$ time using the same approach as for merging. If $b$
is to be removed, we first remove the rank path of $L(a)$ and identify the rank tree $T_b$ containing $b$. We delete the
path from $b$ to the root of $T_b$, thereby partitioning this rank tree into $O(\log n)$ smaller rank trees, sorted by
ranks. We pair up rank trees as described above and add a new rank path to form the updated $L(a)$.
All of this can be done in $O(\log n)$ time.

\subsection{Maintaining bitmaps}\label{subsec:MaintainBitmaps}
Finally, we need to update integers $n(u)$ for $\mathcal C$-nodes $u$ as well as the $\edge$-bitmaps. The former is
done exactly as in Sections~\ref{subsec:Insertion} and~\ref{subsec:Deletion} so let us focus on the bitmaps.
If a level $i$-edge $e$ is removed, we do the following for each of its endpoints $a$. In the leaf $a$ of $\mathcal C_L$,
we check in $O(\log\log n)$ time if $e$ was the only level $i$-edge incident to $a$. If so, we set
$\edge(a)[i]\leftarrow 0$ and we move up $\mathcal C_L$, updating the bitmap of the current
node as the bitwise 'or' of its children. Since $\mathcal C_L$ has $O(\log n)$ height, only $O(\log n)$
bitmaps need to be updated. Similarly, if $e$ is added, we set $\edge(a)[i]\leftarrow 1$ and
update bitmaps for ancestors in the same way. For nodes of $\mathcal C_L$ whose children change, we also update
their bitmaps bottom up by taking the bitwise 'or' of their children. Only $O(\log n)$ nodes are affected in
$\mathcal C_L$ after an update in $\mathcal C$ so total time is $O(\log n)$.

This completes the description of the first version of our
data structure. Correctness follows since the data structure is a simple
variation of that of Thorup where spanning trees of clusters are not kept; rather, our search procedures
certify that a spanning tree exists for an explored component and this suffices to maintain the cluster forest.
From the analysis above, our data structure handles updates in $O(\log^2n)$ amortized
time and queries in $O(\log n)$ time. In the next section, we speed up both bounds by a factor of $\log\log n$.

\section{An Improved Data Structure}\label{sec:ImprovedDataStruct}
In this section, we give our improved data structure. Before going into details, let us highlight
the main differences between this structure and that of the previous section. One ingredient is to add shortcuts to $\mathcal C_L$. Each shortcut skips $O(\log\log n)$ nodes and this will allow our search procedures to walk up trees of $\mathcal C_L$ in
$O(\log n/\log\log n)$ time per traversal when identifying visited nodes of a multigraph $M_i$.
Adding these shortcuts essentially corresponds to turning the forest $\mathcal C_L$ of binary
trees into one consisting of trees with a branching factor of order $\log n$, and reducing the
height of the trees to order $\log n/\log\log n$.
Furthermore, we will modify $\mathcal C_L$ by using \emph{lazy} local trees similar
to those of Thorup~\cite{Thorup00} instead of the local trees presented in the previous section.
This is done to maintain $\mathcal C_L$ more efficiently during changes. Unfortunately,
Thorup's lazy local trees increase the height of $\mathcal C_L$ to $O(\log n\log\log n)$ so our shortcuts
will not give any speed-up over the data structure in the previous section. Instead, we shall use a slightly more
complicated type of lazy local tree which has the properties we need while keeping the height of $\mathcal C_L$ bounded
by $O(\log n)$. The idea is to partition the children of each node of $\mathcal C$ into so called
heavy children and light children and construct the lazy local tree only for the light
children and the local tree of the previous section for the heavy children. Balancing this
in the right way will ensure a logarithmic depth of trees while still getting
the speed-up from lazy local trees. Finally, we will need another system of shortcuts for quickly identifying edges to be
explored by the search procedures; Thorup uses a similar system but it does not fit into our framework as our lazy local
trees are different from his. As shown in Lemma~\ref{Lem:Iterator} in Section~\ref{subsec:Iterator}, with these shortcuts, the search procedures can visit edges in only
$O(\log n/\log\log n)$ time per edge plus an additive $O(\log n)$ if we are in case $1$ in
Section~\ref{subsec:Deletion}; note that the latter can be paid for by the \texttt{delete} operation since a replacement path for the deleted edge has been found. In
Section~\ref{subsec:IndShortcuts}, we define these shortcuts and in
Section~\ref{subsec:Iterator}, we give an algorithm that uses these shortcuts to explore level $i$-edges;
we refer to it as a \emph{level $i$-iterator} or just \emph{iterator}.

\subsection{Lazy local trees}\label{subsec:Lazy}
Thorup~\cite{Thorup00} introduced lazy local trees and showed how they can be maintained more
efficiently than the local trees in Section~\ref{subsec:LocalTrees}. Let $u$ be a non-leaf
node of $\mathcal C$ and let $L$ be the set of children of $u$. To form the lazy local tree of
$u$, $L$ is divided into groups each of
size at most $2(\log n)^\alpha$, where $\alpha$ is a constant that we may pick as large as we
like. The nodes in each group are kept in a BBST ordered by $n(v)$-values. One of these trees is
the \emph{buffer tree} while the others are the \emph{bottom trees}. The root of a bottom tree
has rank equal to the maximum rank of its leaves. These bottom trees
are paired up to form at most $\log n$ rank trees, as described in
Section~\ref{subsec:LocalTrees}. The roots of the rank trees together with the root of the
buffer tree are leaves of a BBST called the \emph{top tree} where leaves
are ordered according to rank (and the root of the buffer tree is regarded as, say, the smallest
element). Together, these trees form the lazy local tree of $u$ which is rooted at the root of
the top tree. Note that bottom, buffer, and top trees have
polylogarithmic size only. It is ensured by the data structure of Thorup that for each bottom
tree $B$, new leaves are never added to $B$ and ranks of leaves in $B$ are not changed. This will also hold for our data structure.

We shall use these lazy local trees to improve update time to $O(\log^2n/\log\log n)$.
However, it is easy to see that due to the BBSTs in lazy local trees, if we form $\mathcal C_L$
with these trees, the depth of $\mathcal C_L$ becomes
$O(\log n\log\log n)$. If we use the same approach as in the previous
section, we thus increase query and update time by a factor of $\log\log n$. Adding shortcuts
to $\mathcal C_L$ will avoid this slowdown but this gives a data structure with the same bounds
as in the previous section.

To handle this, we introduce a new type of lazy local trees. Let $u$ be a non-leaf node of $\mathcal C$. A child $v$ of $u$ in $\mathcal C$
is said to be \emph{heavy} if $n(v)\geq n(u)/\log^\epsilon n$ and otherwise it is \emph{light}; here $\epsilon > 0$ is
a constant that we may pick as small as we like.

Our lazy local tree $L(u)$ of $u$ is illustrated in Figure~\ref{fig:LazyLocalTree}.
\begin{figure}
\centerline{\scalebox{1.0}{\input{LazyLocalTree.pstex_t}}}
\caption{Lazy local tree $L(u)$. Rank trees are black, top and bottom trees are white, the buffer tree is grey, nodes of $L(u)\cap\mathcal C$
are white, and nodes of $L(u)\setminus\mathcal C$ are grey.}
\label{fig:LazyLocalTree}
\end{figure}
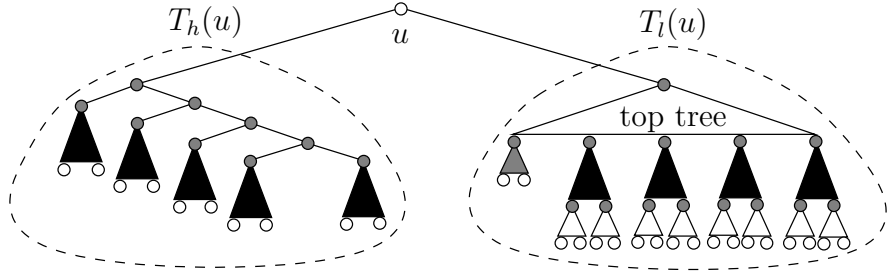
It is rooted at $u$ and has
two children. One child is the root of a tree $T_h(u)$ having the
heavy children of $u$ as leaves and the other child is the root of a tree $T_l(u)$ having the light children of $u$ as
leaves (to simplify the description, we assume that $u$ has both light and heavy children; if there were no light
(resp.~heavy) children, we would simply identify $L(u)$ with $T_h(u)$ (resp.~$T_l(u)$)).
We call $T_h(u)$ the \emph{heavy tree} (of $u$) and it is defined as the local tree from the previous section over
the heavy children of $u$; note that its size is asymptotically bounded by the number of heavy children of $u$
which is at most $\log^\epsilon n$.
The tree $T_l(u)$ is called the \emph{light tree} (of $u$) and it is Thorup's lazy local tree for the set of light children of $u$.

The following lemma shows that when $\mathcal C_L$ is formed from
$\mathcal C$ by inserting these lazy local trees, the height of trees in $\mathcal C_L$ is only a constant factor larger than
that in Section~\ref{sec:SimpleDataStruct}.
\begin{Lem}\label{Lem:HeightCL}
The height of $\mathcal C_L$ (with lazy local trees) is $O(\log n)$.
\end{Lem}
\begin{proof}
Let $u$ be a non-leaf node in $\mathcal C$ and let $v$ be one of its children, also in $\mathcal C$. If
$v$ is a heavy child of $u$ then $v\in T_h(u)$ so its depth in
$L(u)$ is at most $\log(n(u)/n(v)) + 1$. Now assume that $v$ is a light child of $u$. Then $v$ belongs to either
a bottom tree or the buffer tree of $L(u)$. In the latter case, the depth of $v$ in $L(u)$ is $O(\log\log n)$.
In the former case, let $B$ be the bottom tree containing $v$ and let $w$ be a leaf of $B$ maximizing $n(w)$.
By definition, the root of $B$ has rank $\rank(w)$. Hence the depth of $v$ in $L(u)$ is at most
$\log(n(u)/n(w)) + O(\log\log n)\leq\log(n(u)/n(v)) + O(\log\log n)$. This is $O(\log(n(u)/n(v)))$ as
$n(v) < n(u)/\log^\epsilon n$ implies $\log\log n = O(\log(n(u)/n(v)))$.

It follows that in both cases, $v$ has depth $O(1 + \log(n(u)/n(v)))$ in $L(u)$. The height of $\mathcal C$ is at most
$\log n$ so by a telescoping sums argument, $\mathcal C_L$ has height $O(\log n)$, as desired.
\end{proof}

\subsection{Maintaining lazy local trees}\label{subsec:MaintainLazy}
Now, let us describe how to maintain lazy local trees of $\mathcal C_L$ corresponding to changes in $\mathcal C$.
For technical reasons, we assign ranks to rank path nodes in heavy trees by $\rank(v_i) = \rank(r_i)$ for $i = 1,\ldots,k-1$, where $v_i$ and $r_i$ are defined as in Section~\ref{subsec:LocalTrees}.
In the following, \emph{rank nodes} are nodes that are assigned a rank. Note that every node of $\mathcal C_L$ is
a rank node except non-leaf nodes of a buffer or top tree and nodes of a bottom tree $B$ that are neither leaves
nor the root of $B$.

\subsubsection{Merging}\label{subsubsec:Merging}
We need to maintain lazy local trees when merging $\mathcal C$-nodes and when
adding and removing children from $\mathcal C$-nodes. We start with merging.
Consider two $\mathcal C$-nodes $u$ and $v$ that are to be merged into $u$. Denote by $u'$ the
updated $u$ after the merge. Note that every heavy child of $u'$ must be a heavy child of either $u$ or $v$.
Hence, we can form $T_h(u')$ by traversing every leaf $w$ of $T_h(u)$ and $T_h(v)$ and adding it as a leaf of $T_h(u')$
iff $n(w)\geq n(u')/\log^\epsilon n$. Total time for this is $O(|T_h(u)| + |T_h(v)|) = O(\log^\epsilon n)$.

Forming $T_l(u')$ is done as in~\cite{Thorup00}. The two buffer trees
are merged in time bounded by the smaller of the two trees. If the number of leaves of the merged buffer tree exceeds
$\log^\alpha n$,
it is turned into a bottom tree for $u'$, leaving an empty buffer tree. The root of the new bottom tree is paired up
with other rank nodes, if needed. We can pay for all buffer tree merges by giving a node $(\log\log n)^2$ credits
when it switches from not belonging to belonging to a buffer tree: every time it is moved to another buffer tree, we spend
$O(\log\log n)$ time for the node and the size of the buffer tree containing the node grows by a factor of at least $2$;
hence the node is moved at most $O(\log\log n)$ times before either being deleted or being moved to a bottom tree.

We propose a different approach for merging top trees than that of Thorup; we feel ours is simpler as it avoids
keeping a special bitmap associated with each top tree node.
Assume w.l.o.g.~that the top tree $T(u)$ in $T_l(u)$ is no bigger than the top tree $T(v)$ in $T_l(v)$.
For each leaf of $T(u)$ we binary search for a leaf with the same rank in $T(v)$ in $O(\log\log n)$ time. While there are
roots with equal ranks, we pair them up as before, finally obtaining the top tree for $T_l(w)$. The time for this is
$O(|T(u)|\log\log n)$ plus time bounded by the number of new rank nodes created. Below we will show how the creation of
rank nodes are paid for when they are deleted. With an amortized analysis similar to that above for buffer trees,
we can pay for all top tree
updates if we assign $(\log\log n)^2$ credits to a node when it switches from not appearing to appearing as a leaf in a top
tree, and if we borrow $(\log\log n)^2$ credits whenever we delete a leaf of a top tree (thereby borrowing from a new rank node)
and distribute these borrowed credits evenly among the remaining leaves.

Since $n(u') = n(u) + n(v)$, we may have some leaves $w\in T_h(u)\cup T_h(v)$ with $n(w) < n(u')/\log^\epsilon n$
and hence $w$ should belong to $T_l(u')$. All such nodes are added to the buffer tree;
as before, this tree is turned into a bottom tree if it gets more than $\log^\alpha n$ leaves. Total time for this is
$O(\log^\epsilon n\log\log n)$ (plus time bounded by the number of new rank nodes created)
since at most $2\log^\epsilon n$ nodes need to be moved from $T_h(u)\cup T_h(v)$.

Finally, let $p$ be the parent of $u$ and $v$ in $\mathcal C$; below we will add shortcuts that
allow us to identify $p$ from $u$ (equivalently from $v$) in $O(\log n/\log\log n)$ time (Lemma~\ref{Lem:Shortcut}).
We remove $u$ and $v$ as leaves of $L(p)$ and then we add $u'$ as a leaf of $T_h(p)$ if $n(u')\geq n(p)/\log^\epsilon n$ and
otherwise we add $u'$ as leaf of the buffer tree in $T_l(p)$. This takes $O(\log\log n)$ time
since both the buffer tree and $T_h(p)$ have poly-logarithmic size and their roots have depth
$O(\log\log n)$ in $L(p)$.

\subsection{Removing a child}\label{subsec:RemovingChild}
Consider removing a $\mathcal C$-node
child $v$ of a $\mathcal C$-node $u$ and adding it as a child of a new $\mathcal C$-node $w$ which is added as a child of the
parent $\mathcal C$-node $p$ of $u$. We first focus on removing $v$ and we let $u'$ denote $u$ after this update.

Assume first that $v\in T_h(u)$. After removing $v$ from $T_h(u)$, we have $n(u') = n(u) - n(v)$ and hence
some nodes may need to be moved from $T_l(u)$
to $T_h(u)$ in order to form $L(u')$. Identifying such nodes in the buffer tree of $T_l(u)$ can be done in $O(\log\log n)$
time per node. Now suppose $w$ is a leaf of a bottom tree $B$ that needs to be moved to $T_h(u)$. Let $b$ be the root of $B$ and let
$a$ be the leaf of the top tree of $T_l(u)$ having $b$ as descendant. Since
$\rank(b)\geq \rank(w)\geq\lfloor\log(n(u')/\log^\epsilon n)\rfloor > \lfloor\rank(a) - \epsilon\log\log n\rfloor$
and since ranks are strictly decreasing on the path from
$a$ to $b$, $b$ has depth at most $\lceil\epsilon\log\log n\rceil$ in the subtree of $T_l(u)$ rooted at $a$. A DFS
from $a$ identifies all bottom tree roots with at most this depth in $O(\log^\epsilon n)$ time; let $L_a$ be the set of
leaves in these bottom trees that need to be moved to $T_h(u)$. Using binary search in the $O(\log^\epsilon n)$ bottom
trees, we identify $L_a$ and move it to $T_h(u)$ in time $O((\log^\epsilon n + |L_a|)\log\log n)$.
Since all rank leaves of the top tree of $T_l(u)$ have distinct ranks, there are only $O(\log\log n)$ choices for $a$.
Also, the maximum number of leaves to be moved is bounded by the number of leaves of $T_h(u')$ which is at most
$\log^\epsilon n$. Total time is thus $O(\log^\epsilon n(\log\log n)^2)$.

We may also need to move $u$ in $L(p)$.
If $u$ belongs to a bottom tree in $T_l(p)$, we move it to the buffer tree as we do not allow ranks
of leaves in bottom trees to change. We also move $u$ to the buffer tree if $u\in T_h(p)$ and $n(u') < n(p)/\log^\epsilon n$.
As we saw for merge, the time for this is $O(\log n/\log\log n)$.

As in~\cite{Thorup00}, we need to do more global updates whenever removing a leaf from a bottom tree $B$ of $T_l(u)$ reduces
the maximum rank of leaves in $B$ and hence the rank of the root $b$ of $B$. We use a similar approach and amortized
analysis as Thorup here: first, delete all rank nodes from $b$ to the ancestor leaf $a$ of the top tree of $T_l(u)$. Then
pair nodes of equal rank as before. For the amortized analysis, we can assume that the graph $G$ ends with no edges so all rank nodes
end up being deleted and we can amortize creation of rank nodes in $T_l(u)$ over deletion of rank nodes in $T_l(u)$.
A rank node is only deleted
when a bottom tree root has its rank reduced. Since a rank is at most
$\log n$, a particular bottom tree can have its root rank reduced at most $\log n$ times (nodes are never added to a bottom
tree and ranks of bottom tree leaves do not change) so in total it gives rise to at most $\log^2n$ rank node deletions for that bottom tree.
But since a bottom tree starts out with $(\log n)^\alpha$ leaves that will all be
deleted eventually, we can amortize each rank node deletion over $(\log n)^{\alpha - 2}$ deletions of bottom tree leaves.
When removing child $v$, we delete at most $\log^\epsilon n$ leaves from bottom trees of $T_l(u)$ in order to form
$L(u')$ so we can amortize each rank node
deletion over $(\log n)^{\alpha - 2 - \epsilon}$ deletions of children in $\mathcal C$. Hence if we pick constant
$\alpha \geq 2 + \epsilon$, we can afford to pay for rank node deletions and also to pay for the $(\log\log n)^2$ credits
that may have been borrowed from a rank node during a merge.

It remains to consider the case $v\in T_l(u)$. Above we showed how to efficiently remove up to order $\log^\epsilon n$ leaves from
$T_l(u)$ so clearly the single leaf $v$ can also be removed efficiently. We then move additional leaves from
$T_l(u)$ to $T_h(u)$ and move $u$ to $T_h(p)$ or to the buffer tree of $T_l(p)$, as above.

\subsection{Adding a child}
Now consider adding $v$ as a child of $w$ and $w$ as a child of $p$. The former is trivial as $w$ has
no children before adding $v$. If $n(v)\geq n(p)/\log^\epsilon n$, we add $w$ to $T_h(p)$ and otherwise we add it to the buffer tree of
$T_l(p)$. Given $p$, total time for this is $O(\log\log n)$.

\subsection{Shortcutting}\label{subsec:Shortcutting}
In order to get our $\log\log n$ speed-up for updates and queries, we need to be able to traverse $\mathcal C_L$ faster.
Thorup~\cite{Thorup00} introduced a system of shortcuts for quickly identifying certain edges incident to clusters.
This will not suffice in our approach since for our search procedures, we also need to move quickly from a leaf
of $\mathcal C_L$ to its ancestor level $i$ node in order to identify the associated level $i$ cluster, for some $i$.
We therefore introduce a different system of shortcuts in the following. To avoid skipping past a level $i$ node
with these shortcuts,
our data structure associates, for each node of a heavy tree $T_h(u)$, the level $\ell(u)$ of $u$.
We can easily extend the data structure for lazy local trees to maintain these values within the same time bound since
$T_h(u)$ has only size $O(\log^\epsilon n)$.

Let us color each node of $\mathcal C_L$ either white or black. For the coloring below, we
define a \emph{black-induced child}
of a node $u\in \mathcal C_L$ to be a black descendant $v$ of $u$ such that all interior nodes on the path from $u$ to
$v$ in $\mathcal C_L$ are white. If $u$ is black, we add a shortcut between $u$ and each of its black-induced children. The shortcut is directed to $u$, allowing us
to move quickly up in $\mathcal C_L$.
The black-induced parent of a node is defined similarly. Note that the shortcuts (with directions reversed) form a forest of rooted trees
over the black nodes.

Now, let us define the coloring of nodes of $\mathcal C_L$. The following nodes are colored black:
\begin{enumerate}
\item every $\mathcal C$-node $u$ with $\ell(u) = i\lfloor\epsilon\log\log n\rfloor$ for some integer $i$,
\item every rank node $u$ having a parent rank node $v$ in $\mathcal C_L$ such that
      $\rank(u)\leq i\lfloor\epsilon\log\log n\rfloor < \rank(v)$ for some integer $i$,
\item every leaf of $\mathcal C_L$ and of every buffer, bottom, and top tree, and
\item every node of a buffer, bottom, and top tree whose depth in that tree is divisible by
      $\lfloor\epsilon\log\log n\rfloor$ (in particular, every root of such a tree is black).
\end{enumerate}
A black node is of \emph{type} $1$, $2$, $3$, and/or $4$, depending on these four cases.
All other nodes are colored white. For performance reasons, we shall only maintain $\edge$-bitmaps for black nodes. Lemma~\ref{Lem:Shortcut} below shows that
these shortcuts give a $\log\log n$ speed-up when moving up a tree of
$\mathcal C_L$. We first need the following result.
\begin{Lem}\label{Lem:RanksCL}
Ranks are non-decreasing along any simple leaf-to-root path in $\mathcal C_L$. Between any two consecutive
$\mathcal C$-nodes on such a path, there are at most two pairs of nodes of equal rank.
\end{Lem}
\begin{proof}
The first part of the
lemma will follow if we can show that ranks are non-decreasing along any simple leaf-to-$u$ path $P$ in a lazy local
tree $L(u)$. This is clearly the case for leaves in $T_h(u)$. A leaf $v$ in $T_l(u)$ either belongs to a bottom or buffer
tree $T$. Assume the former since otherwise, $u$ and $v$ are the only rank nodes on the $v$-to-$u$ path $P$
and since both are $\mathcal C$-nodes, $\rank(u)\geq \rank(v)$.

Only the first node $v$ and last node $r$ of subpath $P[v,r] = T\cap P$ are rank nodes. Since $r$ is the root of $T$,
it has maximum rank among leaves in $T$ so $\rank(v)\leq\rank(r)$. Let $l$ be the leaf of the top tree of
$L(u)$ belonging to
$P$. All nodes on $P[r,l]$ belong to a rank tree so ranks are increasing along this subpath. For the subpath
$P[l,u]$, only $l$ and $u$ are rank nodes. Let $L$ be the set of leaves of $T_l(u)$ formed by picking
one of maximum rank from each bottom tree descending from $l$. Then
$\rank(l)\leq\lfloor\log(\sum_{u'\in L}n(u'))\rfloor\leq\lfloor\log n(u)\rfloor = \rank(u)$. This shows the first part of
the lemma.

For the second part, let $u$ and $v$ be $\mathcal C$-nodes where $v$ is a child of $u$. Assume first
that $v$ is a leaf of $T_h(u)$. Ranks are strictly increasing on the path from $v$ to the root $r$ of the rank tree
containing $v$. Ranks are also strictly increasing along the rank path in $T_h(u)$. Hence, there
are at most two pairs of equal rank nodes in $\mathcal C_L$
between $u$ and $v$, namely $r$ and its parent and the root of $T_h(u)$ and
$u$. Now consider the case where $v$ is a leaf of $T_l(u)$ and again assume it belongs to a bottom tree $B$. Let $r$ be the root
of $B$ and
let $l$ be the leaf of the top tree which is an ancestor of $v$. Then again, since ranks are strictly increasing along
any leaf-to-root path in a rank tree, there can be at most two equal-rank pairs between $u$ and $v$, namely
$(v,r)$ and $(l,u)$. This completes the proof.
\end{proof}

\begin{Lem}\label{Lem:Shortcut}
Given $\mathcal C_L$ with shortcuts, given a level $i$, and given a $\mathcal C$-node of $\mathcal C_L$ with
an ancestor level $i$ node, we can identify this ancestor in $O(\log n/\log\log n)$ time.
\end{Lem}
\begin{proof}
Let $v$ be the given node. To identify the ancestor level $i$ node $u$ of $v$, we start by
traversing the $v$-to-root path of the tree in $\mathcal C_L$ containing $v$ and we stop if we reach $u$ or a black
node. Since $v$ is a $\mathcal C$-node and since light trees have black leaves, all nodes visited are rank nodes.
Lemma~\ref{Lem:RanksCL} then implies that we visit at most $O(\log\log n)$ nodes before stopping. Hence, the
traversal takes $O(\log\log n)$ time. Assume that we reach a black node $b_1$ as we are done if we reach $u$.

From $b_1$ we traverse shortcuts until we get to the lowest-depth black node $b_2$ having $u$ as ancestor.
Finally we traverse the $b_2$-to-root path in $\mathcal C_L$ until we reach $u$. The latter takes $O(\log\log n)$ time
by an argument similar to the above.

Next we show that there are $O(\log n/\log\log n)$ shortcuts between $b_1$ and $b_2$. Since $\mathcal C$ has logarithmic height, there are only
$O(\log n/\log\log n)$ shortcut endpoints of type $1$.
Lemma~\ref{Lem:RanksCL} implies the same bound for shortcut endpoints of
type $2$.
If a shortcut ends at a type $3$ node $b$ which is a leaf of a buffer or bottom tree, it means that we enter a light
tree $T_l(w)$. We encounter only one additional type $3$ node in $T_l(w)$, namely a leaf of a top tree. Since
$n(w) > n(b)\log^\epsilon n$ we have $\rank(w) > \rank(b) + \lfloor\epsilon\log\log n\rfloor$ and
since a rank is at most $\log n$, Lemma~\ref{Lem:RanksCL} implies that
we encounter no more than $O(\log n/\log\log n)$ type
$3$ nodes between $b_1$ and $b_2$. Finally, this bound on the number of type $3$ nodes and Lemma~\ref{Lem:HeightCL} give
the same asymptotic bound on the number of type $4$ nodes.

What remains is to describe how to avoid jumping past $u$ when traversing the shortcuts. Let $(b_2,b_3)$ be the shortcut
that jumps past $u$, if any. Since leaves of light
trees are black, $b_3$ must belong to some heavy tree $T_h(w_{b_3})$. If $b_2$ belongs to a light tree, it
must belong to the root of the top tree in $T_l(u)$ since that root is black. We can avoid this case as follows: whenever
we reach the root of a top tree, its parent is a $\mathcal C$-node and we compare its level with $i$ to determine
whether we should continue with the shortcuts.

Now, consider the case where $b_2$ belongs to a heavy tree $T_h(w_{b_2})$. Recalling that for every node of a heavy tree
$T_h(w)$ we keep the level $\ell(w)$ of $w$, we can check that $\ell(w_{b_2})\leq i < \ell(w_{b_3})$ to detect that $a$
is the last node that we should visit with shortcuts. This completes the proof.
\end{proof}

\subsection{Induced shortcuts}\label{subsec:IndShortcuts}
Lemma~\ref{Lem:Shortcut} allows us to speed up part of our search procedure, namely identifying the endpoints
(level $(i+1)$ clusters) of an edge in a multigraph $M_i$ from the endpoints of the corresponding edge in $G$; we can
do this in $O(\log n/\log\log n)$ time per endpoint. We also need a faster
iterator for level $i$-edges incident to explored level $(i+1)$ clusters. We focus on this in the following.

Define an \emph{$i$-induced forest $\mathcal F_i$} as in~\cite{Thorup00}: its $i$-induced leaves are the leaves
of $\mathcal C_L$ with an incident level $i$-edge. Its $i$-induced roots are the level $(i+1)$ nodes of $\mathcal C_L$
having descendant $i$-induced leaves. Its $i$-induced branch nodes are the nodes of $\mathcal C_L$ with both
children have descending $i$-induced leaves. The $i$-induced parent of an $i$-induced node is its nearest $i$-induced
ancestor. This defines $\mathcal F_i$.

A straightforward level $i$-iterator performs a DFS in a tree of $\mathcal F_i$. However, maintaining the
edges of $\mathcal F_i$ will be too expensive. Instead, we introduce a new system of shortcuts in $\mathcal C_L$ that
will allow
the DFS to move between any two incident $i$-induced nodes of $\mathcal F_i$ in $O(\log n/\log\log n)$ time. Since
a tree of $\mathcal F_i$ is binary, the number of branch nodes of $T$ is bounded by the number of leaves of $T$, so
this will give a level $i$-iterator with $O(\log n/\log\log n)$ amortized time per level $i$-edge. In the following,
we define the new shortcuts. Refer to the following types of nodes of $\mathcal C_L$ as \emph{special}:
\begin{enumerate}
\item every $\mathcal C$-node $u$ with
      $\ell(u) = i\lfloor\log\log n\rfloor\lfloor\epsilon\log\log n\rfloor$ for some integer $i$,
\item every leaf of $\mathcal C_L$, and
\item every rank node $u$ of a light tree with $\rank(u) = i\lfloor\log\log n\rfloor\lfloor\epsilon\log\log n\rfloor$ for some integer $i$.
\end{enumerate}
Note that every special node is black. Also note that we defined type $3$ special nodes using
equality rather than inequality as for type $2$ black nodes. This suffices since ranks increase by
$1$ as we move up rank nodes of a light tree; this is not the case in heavy trees where ranks can
increase by larger values along a rank path. For a special node $u$, define a \emph{special child} of $u$ to be
a descendant special node $v$ such that all nodes between $u$ and $v$ are not special. Special parents are defined similarly.
For any level $i$, if there is a unique special child $v$ of $u$ for which $\edge(v)[i] = 1$, we add a
shortcut (bidirected pointer) between $u$ and $v$. To distinguish these shortcuts from those of
Section~\ref{subsec:Shortcutting}, we refer to the former as \emph{$i$-induced shortcuts} or just \emph{induced shortcuts}
and the latter as
\emph{standard $i$-shortcuts} or just \emph{standard shortcuts}. Observe that for all $i$-induced shortcuts $(a,b)$, where $b$ is a special child
of $a$, there is an edge in $\mathcal F_i$ from $a$ or an ancestor of $a$ to $b$ or a descendant of $b$.
For each special node $u$, we keep a BBST with a leaf for each $i$ containing the $i$-induced
shortcuts to a special child and/or parent (if they exist).

\subsection{Faster iterator}\label{subsec:Iterator}
Now let us present the level $i$-iterator. It starts at the root $v$ of a tree in $\mathcal F_i$, i.e., $v$ is a level
$(i+1)$ node of
$\mathcal C_L$. It performs a DFS of the subtree of $\mathcal C_L$ rooted at $v$ with the following modification: if
it visits a black node $w$ for which $\edge(w)[i] = 0$, it backtracks; if it visits a special node
$w'$ with an $i$-induced shortcut to a special child, it visits this special child instead of the children of $w'$
in $\mathcal C_L$. When it reaches a leaf $l$ of $\mathcal C_L$, it identifies the group of incident level $i$-edges
with a binary search in the BBST associated with $l$ and then iterates over these edges.
This completes the description of the level $i$-iterator.
Lemma~\ref{Lem:Iterator} below shows the performance of the level $i$-iterator.
To prove it, we need two additional lemmas.
\begin{Lem}\label{Lem:WhiteTree}
Any node of $\mathcal C_L$ has only $O((\log n)^{3\epsilon})$ black-induced children.
\end{Lem}
\begin{proof}
Let $u$ be a node of $\mathcal C_L$. If $u$ is a non-leaf node of a bottom, buffer, or top tree, the lemma follows
from the definition of type $3$ and $4$ black nodes. Otherwise, $u$ is a rank node. For
any black-induced child $v$ of $u$, Lemma~\ref{Lem:RanksCL} and the definition of type $1$ and $2$ black nodes imply
that $v$ has depth at most $3\epsilon\log\log n$ in the subtree of $\mathcal C_L$ rooted at $u$. As $\mathcal C_L$
is binary, the lemma follows.
\end{proof}
\begin{Lem}\label{Lem:SpecialDist}
For any $i$-induced shortcut $(a,b)$, the simple $a$-to-$b$ path in $\mathcal C_L$ has length $O((\log\log n)^4)$.
\end{Lem}
\begin{proof}
Let $P$ be the simple $a$-to-$b$ path in $\mathcal C_L$. Clearly, $P$ contains only $O((\log\log n)^2)$
$\mathcal C$-nodes. Let $P'$ be a subpath of $P$ containing no $\mathcal C$-nodes. Then $P'$ is either contained
in a heavy or a light tree. In the former case, $|P'| = O(\log\log n)$. In the latter case, we encounter at
most $O(\log\log n)$ nodes of buffer, bottom, and top trees on $P'$. Since consecutive rank nodes of $P'$ differ in rank
by exactly $1$ (as they all belong to a light tree and hence to a rank tree), we encounter at most
$O((\log\log n)^2)$ rank nodes on $P'$ so $|P'| = O((\log\log n)^2)$.
\end{proof}
\begin{Lem}\label{Lem:Iterator}
The level $i$-iterator above traverses a tree in $\mathcal F_i$ with $k$ leaves in $O(k\log n/\log\log n)$ time. The
time to visit the first $k'$ leaves is $O(k'\log n/\log\log n + \log n)$.
\end{Lem}
\begin{proof}
Correctness follows easily from the definition of $\edge$-bitmaps, $i$-induced shortcuts, and BBSTs associated with
leaves of $\mathcal C_L$ so let us focus on the time bound to traverse a $k$-leaf tree $T$ in
$\mathcal F_i$. Let $T_L$ be the tree in $\mathcal C_L$ obtained by replacing each edge $(a,b)\in T$ with the
corresponding simple path $P$ in $\mathcal C_L$ between $a$ and $b$. By
Lemmas~\ref{Lem:HeightCL} and~\ref{Lem:RanksCL}, there can only be
$O(\log n/(\log\log n)^2)$ special nodes on such a path $P$. Hence, since $T$ has no degree $2$-vertices, the total
number of special nodes and hence $i$-induced shortcuts traversed by the level $i$-iterator in $T_L$ is
$O(k\log n/(\log\log n)^2)$. For each special node visited, $O(\log\log n)$ time is spent on binary search to find the
next $i$-induced shortcut, if it exists. Hence, the total time spent on visiting special nodes and traversing
$i$-induced shortcuts is $O(k\log n/\log\log n)$.

We will now show that the number of additional nodes visited by the DFS is $O(k(\log n)^{3\epsilon}(\log\log n)^4)$. Since
only constant time is spent for each such node, this will show the first part of the lemma.
First we bound the number of visited nodes of $T_L$ which are not special.
Let $(a,b)$ and $P$ be as above. If we traverse $P$ from $a$ then it follows from Lemma~\ref{Lem:SpecialDist} that
after at most $O((\log\log n)^4)$ nodes, we will reach either $b$ or a special node $a'$.
Similarly, if we traverse $P$ from $b$ then after at most $O((\log\log n)^4)$ nodes, we
will reach either $a$ or a special node $b'$. If $a'$ and
$b'$ exist then all nodes visited by the DFS on $P[a',b']$ are special nodes connected by $i$-induced shortcuts.
Summing over all such paths $P$, it
follows that the total number of nodes visited on $T_L$ which are not special is $O(k(\log\log n)^4)$.

Finally, let us bound the number of nodes of $\mathcal C_L$ visited by the DFS which are not on $T_L$. Consider a visited
node $u\in T_L$ and let $v\notin T_L$ be a visited node such that $u$ is the nearest ancestor of $v$ belonging to $T_L$.
Note that there is no $i$-induced shortcut from $u$ to a special child since then the DFS would have traversed
this shortcut instead of the children of $u$ in $\mathcal C_L$. In particular, there are only $O(k(\log\log n)^4)$
choices for $u$. Furthermore, all interior nodes on the simple path from $u$ to $v$ in $\mathcal C_L$ are white since
any black node $w$ would have $\edge(w)[i] = 0$ (as $w\notin T_L$),
meaning that the DFS would have backtracked
before reaching $v$. By Lemma~\ref{Lem:WhiteTree}, there are only $O((\log n)^{3\epsilon})$ choices
for $v$ for each $u$. Hence, the total number of nodes visited which are not on $T_L$ is
$O(k(\log n)^{3\epsilon}(\log\log n)^4)$. This shows the first part of the lemma.

For the second part, consider a partially grown DFS tree $T'$ which has visited $k'$ leaves. For every node of $T'$
having two children, at least one of the two subtrees rooted at the children is fully explored. Hence, $T'$ consists of
a path $P$ from the root of $T$ to a leaf of $T$ with fully explored subtrees attached to $P$. The same argument as
above shows that the total time to explore these subtrees is $O(k'\log n/\log\log n)$. By Lemma~\ref{Lem:HeightCL},
it takes $O(\log n)$ time to explore $P$ (the number of special nodes on $P$ is $O(\log n/(\log\log n)^2)$ so we
only spend a total of $O(\log n/\log\log n)$ time on binary searches for these nodes).
\end{proof}

It follows from Lemma~\ref{Lem:Iterator} that the level $i$-iterator spends $O(\log n/\log\log n)$ amortized time
per edge visited plus additional $O(\log n)$ time if a replacement path was found (if such a path is not found, an entire
tree in $\mathcal F_i$ is traversed). The $O(\log n/\log\log n)$ amortized time per edge
is paid for by the increase in the level of the edge and the $O(\log n)$ time is paid for by the deletion of an edge
in $G$ since at most one replacement path is found for such an edge.

It remains to describe how colors, shortcuts (standard and induced) and
$\edge$-bitmaps are maintained when $\mathcal C$
(and hence $\mathcal C_L$) is updated and when edges of $G$ are added/removed or change level. First we deal with
changes to $\mathcal C$. The following lemma will prove useful.
\begin{Lem}\label{Lem:IncidentInducedShortcuts}
Given $\edge$-bitmaps of black nodes and given a special node $u$, we can find the induced shortcuts between $u$ and
its special parent (if any) in $O(\log n)$ time. For any $i$, we can find the $i$-induced shortcut from $u$ to
a special child or determine that no such shortcut exists in $O((\log n)^{3\epsilon}(\log\log n)^4)$ time.
\end{Lem}
\begin{proof}
We first walk up $\mathcal C_L$ from $u$ to identify its special parent $p$. By Lemma~\ref{Lem:SpecialDist},
this takes $O((\log\log n)^4)$ time.
Then we perform a DFS in the subtree of $\mathcal C_L$ rooted at $p$ and backtrack if we encounter $u$ or a black node
which is not an ancestor of $u$. If any such black node is encountered for which the $i$th bit of its $\edge$-bitmap
is $1$ then we know that there should not be an $i$-induced shortcut between $u$ and $p$. Otherwise there should be iff
$\edge(u)[i] = 1$.
Let $m$ be the bitmap obtained by taking the bitwise 'or' of the $\edge$-bitmaps of visited black nodes not on the
$u$-to-$p$ path. By Lemmas~\ref{Lem:WhiteTree} and~\ref{Lem:SpecialDist}, finding $m$ takes
$O((\log n)^{3\epsilon}(\log\log n)^4)$ time. Now, there is an $i$-induced shortcut between $u$ and $p$ iff
$m[i] = 0$ and $\edge(u)[i] = 1$. Hence, all induced shortcuts between $u$ and
$p$ can be found in $O(\log n)$ time.

To find the $i$-induced shortcut (if any) to a special child of $u$, we make a DFS from $u$ which backtracks when reaching
a black node. Suppose exactly one visited black node $w$ has $\edge(w)[i] = 1$ (otherwise, there is no
$i$-induced shortcut). If $w$ is special, we have identified the $i$-induced shortcut $(w,u)$. Otherwise, we recurse on $w$.
As above, total time for this is $O((\log n)^{3\epsilon}(\log\log n)^4)$.
\end{proof}

\subsection{Structural changes}
Let us now describe how shortcuts and $\edge$-bitmaps are maintained after structural changes to $\mathcal C_L$.
It follows from Lemma~\ref{Lem:WhiteTree} that for each update to $\mathcal C_L$, we can update colors, standard shortcuts
and $\edge$-bitmaps in $O((\log n)^{3\epsilon})$ time. From the results in Section~\ref{subsec:MaintainLazy}, this will not
affect the overall time bound (if we pick constant $\alpha$ sufficiently large). In the following, we thus only consider
updating induced shortcuts.

We shall restrict our attention to structural changes caused by a \texttt{delete}-operation as \texttt{insert}
corresponds to merging two clusters (or none), a type of update that needs to be supported during a \texttt{delete}.

Recall that after a \texttt{delete}-operation, $\mathcal C$ is updated as follows: some children of a node $u$ are removed
and merged into a single node; this node is either added as a child of $u$ (if a replacement path was found) or it is added
as a child of a new node $u'$ which is added as a child of the parent $p$ of $u$ and the process is repeated recursively on
$p$ (if a replacement path was not found). We observe that all
$\mathcal C$-nodes whose children are updated are contained in two
leaf-to-root paths in $\mathcal C$ after the \texttt{delete}-operation has been executed.

There are two types of induced shortcuts that need to be updated, those incident to a type $1$ special node and those
descending from a type $3$ special node and not ascending from a type $1$ special node (see definitions of types in
Section~\ref{subsec:IndShortcuts}).
Below we show how to update the latter.

For the former, it follows from the above
that we only need to focus on type $1$ special nodes on a leaf-to-root path $P$ in $\mathcal C_L$ (there are two paths but they are handled
in the same manner). Let
$u_1,\ldots,u_k$ be the sequence of special nodes as we move up $P$ during the \texttt{delete}-operation (some of them may be new or merged
nodes and hence do not exist before the \texttt{delete}-operation).
By Lemma~\ref{Lem:IncidentInducedShortcuts}, we can find all induced shortcuts descending from $u_1$ in $O((\log n)^{1 + 3\epsilon}(\log\log n)^4)$ time.
When we reach
$u_j$, $j > 1$, we compute induced shortcuts between $u_{j-1}$ and $u_j$. By Lemma~\ref{Lem:IncidentInducedShortcuts},
this takes $O(\log n\log\log n)$ time (including binary searches in the BBSTs of $u_{j-1}$ and $u_j$) for a total of
$O(\log^2n/\log\log n)$ over all $j$ which the \texttt{delete}-operation can pay for. We also compute induced shortcuts descending from $u_j$ for those
$i$ for which $\edge(u_j)[i] = 1$ and $\edge(u_{j-1})[i] = 0$. Total time over
all $j$ is $O((\log n)^{1 + 3\epsilon}(\log\log n)^4)$ since if $\edge(u_j)[i] = 1$ then $\edge(u_{j'})[i] = 1$ for all $j' > j$,
implying that the second part of
Lemma~\ref{Lem:IncidentInducedShortcuts} is applied at most once for each $i$. Note that all $i$-induced shortcuts from $u_j$ to
a special child which have not been identified by the second part of
Lemma~\ref{Lem:IncidentInducedShortcuts} must have $\edge(u_j)[i] = \edge(u_{j-1})[i] = 1$ and hence must connect $u_j$ to $u_{j-1}$ which we have computed above.
Hence, we correctly compute all induced shortcuts incident to special nodes
on $P$ and we do so within the desired time bound.

Now we show how to maintain induced shortcuts descending from a type $3$ special node $p$ and not ascending from a type $1$ special node. We may assume that
$p$ is not on one of the leaf-to-root paths considered above.
Then the only structural changes to $\mathcal C_L$ that may require such shortcuts to be updated are
\begin{enumerate}
\item a type $3$ special node is created or deleted, or
\item a leaf is removed from a bottom tree in the light tree containing $p$.
\end{enumerate}
A type $3$ special node can pay any polylogarithmic amount when it is created/deleted so consider updating induced shortcuts descending from $p$ when a leaf
$u$ is removed from a bottom tree $B$ in the light tree $T_l(w)$ containing $p$.
To handle this case, we will assume that each rank node of $T_l(w)$ has
$\log n/\log\log n$ credits for each $i$ for which it is an $i$-induced
branch node. To see that this assumption can be made, first observe that when a buffer tree is turned into
a bottom tree, it can pay any polylogarithmic amount if we pick $\alpha$
sufficiently big. This is also the case when a new rank node of $T_l(w)$ is
created/deleted. Since we never add but only remove leaves from
bottom trees, the only other way a rank node of $T_l(w)$ can become an
$i$-induced branch node is if some edge of $G$ has its level increased to $i$. Such an edge can only contribute with two $i$-induced branch nodes to
$\mathcal F_i$ (one for each of its endpoints) so we may add $\log n/\log\log n$
credits to the two new $i$-induced branch nodes which the level increase can
pay for. This shows the desired.

Let $\branch(u)$ be the bitmap where $\branch(u)[i] = 1$ iff there
is an $i$-induced branch node $v\neq u$ on the path from $u$ to $p$.
By Lemmas~\ref{Lem:WhiteTree} and~\ref{Lem:SpecialDist},
we can form $\branch(u)$ in $O((\log n)^{3\epsilon}(\log\log n)^4)$
time which the removal of $u$ from $B$ can pay for.
For each $i$ for which $\branch(u)[i] = 1$, the removal of $u$
may require us to compute an $i$-induced shortcut descending from $p$.
Since a branch node is removed in the process, we can spend its credits to
pay for applying Lemma~\ref{Lem:IncidentInducedShortcuts} to find this
shortcut. We can binary search for each
of the $k$ $1$-entries of $\branch(u)$ in $O(k\log\log n)$ time; for instance, to determine whether the lower half of $\branch(u)$ has any $1$-bits, we can take the bitwise 'and' of $\branch(u)$ and a precomputed bitmap having $1$-bits in the lower half and $0$-bits in the upper half. The 'and' is $1$ iff there are $1$-bits in the lower half of $\branch(u)$.

Now consider an $i$ for which $\branch(u)[i] = 0$. We may assume that a
bitmap $\induced(p)$ associated with $p$ is maintained
where $\induced(p)[j] = 1$ iff
$p$ has a $j$-induced shortcut to a special child. If $\edge(u)[i] = 0$ or
$\induced(p)[i] = 0$, nothing needs to be done for $i$ so assume
$\edge(u)[i] = 1$ and $\induced(p)[i] = 1$. Then we delete the $i$-induced
shortcut descending from $p$. If there is an $i$-induced shortcut from $p$
to its special parent $p'$, we remove it too and recurse on $p'$; the
recursion stops when we reach a special node $q$
without an $i$-induced shortcut to its special parent $q'$. Each induced shortcut can be removed
in $O(\log\log n)$ time using binary search in the associated BBSTs.
By adding $\log\log n$ credits to an induced shortcut when it is created,
we can pay for all deletions of induced shortcuts. If $q'$ is a rank node of
$T_l(w)$, we may need to add an $i$-induced shortcut from $q'$ to a special
child. This can only happen if an $i$-induced branch node between $q$ and
$q'$ disappears and as above, we can spend its credits to pay for finding
this shortcut.

\subsection{Non-structural changes}
Above we dealt with updates of shortcuts due to structural changes in $\mathcal C_L$. We now handle updates when
leaves of $\mathcal F_i$ appear or disappear due to edge level changes.

\paragraph{Edge deletions:}
When a level $i$-edge $e$ is deleted
(possibly due to its level being increased to $i+1$), $\edge(u)[i]$ might change from $1$ to $0$ for
one of its endpoints $u$ which will then no longer be an $i$-induced leaf of $\mathcal F_i$. We describe how
to update other $\edge$-bitmaps accordingly and remove some of the $i$-induced shortcuts. The following is
done for $u$. If there are still level $i$-edges incident to $u$ then no updates are needed.
Otherwise, all $i$-induced shortcuts on the simple path in $\mathcal C_L$ from $u$ to its $i$-induced parent $p$
should be removed. Since $u$ is a leaf of $\mathcal C_L$, it is a special node. We traverse $i$-induced
shortcuts from $u$ to ancestors until we reach a special node $v$ without an $i$-induced shortcut to an ancestor.
Since $u$ is the only $i$-induced leaf below $v$, we delete all shortcuts traversed. We also set $\edge(v)[i] \leftarrow 0$
for all black nodes $v$ between $u$ and $v$ (including $v$) in
$O(\log n/\log\log n)$ time by traversing standard $i$-shortcuts between $u$ and $v$.

We then traverse black nodes up from $v$ in $\mathcal C_L$ and stop if we find a special node $w$. Whenever we visit
a black node $b$, we perform a DFS in the subtree of $\mathcal C_L$ rooted at $b$, backtracking at descending black
nodes. If each black node $b'$ visited below $b$ has $\edge(b')[i] = 0$, $u$ was the only leaf
of $\mathcal C_L$ below $b$ with an incident level $i$-edge so we set $\edge(b)[i]\leftarrow 0$ and proceed
up to the next black node. Conversely, if some black node $b'$ visited below $b$ has $\edge(b')[i] = 1$,
$p$ must be below $b$ and no more bitmaps need to be updated.

Having updated the bitmaps and removed all $i$-induced shortcuts below $v$, we need to check if an $i$-induced
shortcut should be added from $w$ to one of its descendants. By Lemma~\ref{Lem:IncidentInducedShortcuts}, this takes
$O((\log n)^{3\epsilon}(\log\log n)^4)$ time which can be paid for
by the deletion of $e$; here we can also afford to add $\log\log n$ credits
to the shortcut if it was added.

\paragraph{Edge insertions}
Now suppose a level $i$-edge $e$ is inserted. To update $\edge$-bitmaps and add new $i$-induced shortcuts,
we do the following for each endpoint $u$ of $e$. If other level $i$-edges are incident to $u$ then nothing needs to be
done as $u$ is already a leaf of $\mathcal F_i$ so assume otherwise. Inserting $e$ corresponds to updating $\mathcal F_i$
by adding a new edge $(u,p)$, where $p$ is the $i$-induced parent of $u$. Hence, we need to add $i$-induced shortcuts
between $u$ and $p$. Suppose $u$ has an ancestor black node $v$ such that $\edge(v)[i] = 1$. We traverse
standard shortcuts up from $u$ and stop when we identify the first such $v$. Let $u_1,\ldots,u_k$ be the ordered (possibly
empty) sequence of special nodes visited from $u$ to $v$. Since $v$ already has a descending leaf incident to a level $i$-edge
and $u_k$ does not, $p$ must be on the $u_k$-to-$v$ path in $\mathcal C_L$. Hence, the new $i$-induced shortcuts to be added are
$(u_1,u_2),(u_2,u_3),\ldots,(u_{k-1},u_k)$. We also set to $1$ the $i$th bit of the $\edge$-bitmaps of all black nodes visited.
By Lemmas~\ref{Lem:HeightCL} and~\ref{Lem:RanksCL}, all of this can be done in $O(\log n/\log\log n)$ time and we can also afford to add $\log\log n$ credits
to each of the $O(\log n/(\log\log n)^2)$ new $i$-induced shortcuts.

We assumed that a node $v$
with $\edge(v)[i] = 1$ was encountered. If this is not the case, it means that $u$ should not be added to an
existing tree in $\mathcal F_i$. Rather, a new tree should be formed consisting of a single edge $(u,p)$, where $p$ is the level
$(i+1)$-ancestor of $u$ in $\mathcal C_L$. Clearly, the corresponding $i$-induced shortcuts can be added and $\edge$-bitmaps
updated within the same $O(\log n/\log\log n)$ time bound.

We can now conclude with the following theorem.
\begin{theorem}
There is a deterministic data structure for fully dynamic graph connectivity which supports edge insertions/deletions
in $O(\log^2n/\log\log n)$ amortized time per update and connectivity queries in $O(\log n/\log\log n)$ worst case
time, where $n$ is the number of vertices of the graph.
\end{theorem}

\section{Concluding Remarks}\label{sec:ConclRem}
We gave a deterministic data structure for fully-dynamic graph connectivity that achieves an amortized update time of
$O(\log^2n/\log\log n)$ and a worst-case query time of $O(\log n/\log\log n)$, where $n$ is the number of vertices of
the graph. This improves the update time of Holm, de Lichtenberg, and Thorup~\cite{HLT01} and Thorup~\cite{Thorup00} by a factor of
$\log\log n$. We believe our improvement may extend to fully-dynamic minimum spanning tree, $2$-edge, and/or biconnectivity.

There is still a small gap between upper and lower bounds. For instance, for $O(\log n/\log\log n)$ query
time,~\cite{PatrascuDemaine} gives an $\Omega((\log n)^{1 + \epsilon})$ cell-probe lower bound for updates for constant
$\epsilon > 0$. Simultaneously getting $O(\log n)$ update and query time and
improving the $O(\sqrt n)$ worst-case update bound in~\cite{Sparsification} remain the main open problems for fully-dynamic
graph connectivity.

\end{document}

%% file: LazyLocalTree.pstex_t
\begin{picture}(0,0)%
\includegraphics{LazyLocalTree.pstex}%
\end{picture}%
\setlength{\unitlength}{4144sp}%
\begingroup\makeatletter\ifx\SetFigFontNFSS\undefined%
\gdef\SetFigFontNFSS#1#2#3#4#5{%
  \reset@font\fontsize{#1}{#2pt}%
  \fontfamily{#3}\fontseries{#4}\fontshape{#5}%
  \selectfont}%
\fi\endgroup%
\begin{picture}(5266,1615)(3051,-4548)
\put(6710,-3676){\makebox(0,0)[lb]{\smash{{\SetFigFontNFSS{12}{14.4}{\familydefault}{\mddefault}{\updefault}{\color[rgb]{0,0,0}top tree}%
}}}}
\put(4012,-3105){\makebox(0,0)[lb]{\smash{{\SetFigFontNFSS{12}{14.4}{\familydefault}{\mddefault}{\updefault}{\color[rgb]{0,0,0}$T_h(u)$}%
}}}}
\put(6829,-3114){\makebox(0,0)[lb]{\smash{{\SetFigFontNFSS{12}{14.4}{\familydefault}{\mddefault}{\updefault}{\color[rgb]{0,0,0}$T_l(u)$}%
}}}}
\put(5345,-3180){\makebox(0,0)[lb]{\smash{{\SetFigFontNFSS{12}{14.4}{\familydefault}{\mddefault}{\updefault}{\color[rgb]{0,0,0}$u$}%
}}}}
\end{picture}%